\documentclass{article}
\usepackage{graphicx,amsmath,amssymb,mathrsfs,pxfonts} 

\newtheorem{theorem}{Theorem}[section]
\newtheorem{proposition}[theorem]{Proposition}
\newtheorem{corollary}[theorem]{Corollary}
\newtheorem{remark}[theorem]{Remark}
\title{Short-term barrier option price expansion}
\author{Masaaki Fukasawa\thanks{
The University of Osaka\\
\texttt{fukasawa@sigmath.es.osaka-u.ac.jp}\\
\textbf{Acknowledgements:} This work arose from a joint research project involving Mizuho Bank, Ltd., Mizuho-DL Financial Technology Co., Ltd., Mizuho Securities Co., Ltd., and The University of Osaka. The author is grateful to the participants for helpful comments and discussions. The author also thanks Masaya Mitsuhira, Kota Nakazawa, Ryosuke Ohira, and Tamaki Shimooka for their contributions to the numerical experiments.
}}
\date{}

\begin{document}

\maketitle

\begin{abstract}
We derive a short-maturity expansion for up-and-out put barrier option prices under continuous
stochastic volatility when the strike and the barrier approach the spot at the
diffusive scale.  Assuming joint weak convergence of the normalized terminal
return, the relative volatility fluctuation, and the running maximum, together
with uniform integrability, we show that the leading term is the
time-inhomogeneous Black--Scholes barrier price fitted to the forward variance
curve.  The first model-dependent correction is of order
$\theta^{H+1/2}$, where $\theta^H$, $H \in (0,1/2]$, is the order of
the relative volatility fluctuation,
and is represented explicitly through killed Brownian
transition densities.  For regular volatility models, where $H=1/2$, the
coefficient is determined by the short-maturity at-the-money implied-volatility skew.  For 
rough volatility models with $H<1/2$, the coefficient reduces to a one-dimensional
integral.  Numerical experiments show that the
correction materially improves the Black--Scholes approximation.
\end{abstract}


\section{Introduction}

Barrier options are arguably the most basic class of exotic derivatives.  They are among the simplest path-dependent claims, but already exhibit the main difficulty that distinguishes exotic from European option pricing: their value depends not only on the terminal distribution of the underlying asset but also on the behavior of the price path before maturity. Understanding barrier options is therefore a natural first step toward a short-maturity asymptotic theory for more general path-dependent derivatives.

Despite their basic role, rigorous asymptotic results for barrier option prices appear to be relatively scarce compared with the extensive literature on European options.  Closed-form barrier prices in the Black--Scholes model follow from the reflection principle, but explicit formulas are generally unavailable under local or stochastic volatility.  One important line of research uses singular perturbation methods.  Ilhan et al.~\cite{IJS} extended the fast mean-reverting stochastic-volatility approach of Fouque et al.~\cite{FPS} to barrier and other exotic options.  In the presence of an absorbing boundary, the usual Black--Scholes Greek correction acquires an additional boundary contribution.  Fouque et al.~\cite{FJL} subsequently developed a spectral perturbation approach to fast mean-reverting stochastic volatility, obtaining approximations for up-and-out, double-barrier knock-out, and rebate options.  These works demonstrate in particular that the barrier creates a genuine singular-perturbation issue that is absent for European claims.

Other asymptotic approaches are also available.  Hu and Knessl~\cite{HK} obtained a singular-perturbation expansion for barrier prices under the CEV model.  Kato et al.~\cite{KTY} developed a semigroup expansion for the Cauchy--Dirichlet problem associated with continuously monitored barrier options and applied it to regular diffusion-type stochastic volatility models, including SABR.  Their semigroup approach provides a mathematically rigorous justification of the perturbation expansion under suitable regularity conditions.  These works expand with respect to a model parameter while maturity is fixed.

Our concern is different.  We derive a short-maturity expansion without introducing a perturbation parameter into the model.  The barrier itself must then be scaled with maturity.  Indeed, if the barrier is fixed away from the spot, hitting it becomes a rare event as maturity tends to zero.  We instead consider the near-barrier regime in which the distance between the barrier and the spot is of the same order as the typical short-time price fluctuation.  More precisely, for maturity $\theta\downarrow0$, initial variance $v_0$, and fixed normalized levels $b>\max\{k,0\}$, we set
\begin{equation*}
B=S_0(1+\sqrt{v_0\theta},b),\qquad
K=S_0(1+\sqrt{v_0\theta},k).
\end{equation*}
This joint scaling keeps both the strike and the barrier visible in the short-time limit and leads to a Bachelier barrier problem at leading order.  Recent work on the short-time behavior of up-and-in options with a barrier fixed away from the spot~\cite{Bures} concerns the complementary rare-event regime.

The main contribution of this paper is to provide a rigorous short-maturity expansion in this non-degenerate near-barrier regime under minimal assumptions.  In particular, we obtain an explicit remainder $o(\theta^{H+1/2})$ after the first stochastic-volatility correction.  Thus, in contrast with much of the existing barrier asymptotics, where expansions are derived through model-specific perturbation arguments, our result gives a rigorous error estimate under an abstract probabilistic limit assumption and applies both to regular and rough volatility models.  More broadly, because barrier options constitute the most elementary path-dependent exotic claims, the analysis provides a first step toward rigorous short-maturity expansions for more general exotic derivatives.

\section{Barrier option pricing under Bachelier models}
The short-term asymptotic limit of a local stochastic volatility model is shown to be a time-inhomogeneous Bachelier model. 
In this section, we develop a pricing theory of barrier options under the limit model as a preliminary; although the technique here is not very new but anyway we give some details as we could not find an adequate reference with this generality.

Throughout this section, the underlying asset price follows
\begin{equation*}
    \mathrm{d}X_t = \sqrt{v(t)} \mathrm{d}W_t,
\end{equation*}
where $W$ is a standard Brownian motion and $v$ is a deterministic, positive, and locally integrable function. We assume interest rates to be zero for brevity,
\subsection{Up-and-Out Put options}
Let two constants $b \geq k$ be fixed with $b > X_0$, and define the first hitting time
\begin{equation*}
    \tau_b=\inf\{t \geq 0 ; X_t = b\}.
\end{equation*}
The payoff of the Up-and-Out Put (UOP) option with barrier $b$, strike $k$ and maturity $T>0$ is
\begin{equation}\label{payoff}
    (k-X_T)_+ 1_{\{\tau_b > T\}}.
\end{equation}
Let $u(x,t)$ be a solution of the boundary value problem
    \begin{align}
        &\frac{\partial u}{\partial t}(x,t) + \frac{1}{2}v(t)\frac{\partial^2 u}{\partial x^2}(x,t) =0,\ \  (x,t) \in (-\infty,b) \times (-\infty,T), \label{pde}\\
        & u(x,T) = (k-x)_+,  \ \ x \in (-\infty,b), \label{terminal}\\
        & u(b,t) = 0, \ \  t \in (-\infty,T) \label{boundary}
    \end{align}
which is a continuous function in $(x,t)$ on $(-\infty,b] \times (-\infty,T]$.
Then the payoff~\eqref{payoff} 
is represented by
$u(X_{\tau_b \wedge T},\tau_b\wedge T)$.

The unique no-arbitrage price of the option at time $s < \tau_b$ is given by $u(X_s,s)$. Indeed, by It\^o's formula,
\begin{equation*}
    u(X_{\tau_b \wedge T},\tau_b\wedge T) = 
    u(X_s,s) + \int_s^{\tau_b \wedge T}\frac{\partial u}{\partial x}(X_t,t) \mathrm{d}X_t
\end{equation*}
and the second term is the profit-and-loss of a self-financing trading strategy in the underlying asset $X$.

Now we construct a solution to \eqref{pde}-\eqref{boundary}. Let
\begin{equation}\label{wt}
    w(t) = \sqrt{\int_t^T v(r)\mathrm{d}r}.
\end{equation}
The solution of the PDE \eqref{pde} with $b=\infty$ and \eqref{terminal}
is the Bachelier put option price
\begin{equation}\label{Bput}
\begin{split}
     p(x,t)&:=\mathsf{E}[(k-X_T)_+|X_t = x]\\
     & = \int \left(k-x-w(t)z\right)_+\phi(z)\mathrm{d}z\\
     &= \int_{-\infty}^{\frac{k-x}{w(t)}}
     \left(k-x-w(t)z\right)\phi(z)\mathrm{d}z\\ 
     &=(k-x)\Phi\left(\frac{k-x}{w(t)}\right) + w(t)
     \phi\left(\frac{k-x}{w(t)}\right),
\end{split}
\end{equation}
where $\phi$ and $\Phi$ are respectively the standard normal density and distribution functions. 
Then, by the linearity of the PDE \eqref{pde},
\begin{equation}\label{uop}
    u(x,t) := 
    p(x,t) - p(2b-x,t)
\end{equation}
solves \eqref{pde}-\eqref{boundary}.

Let $p_\text{B}$ denote the Bachelier put price function:
\begin{equation*}
    p_{\text{B}}(x,\sigma) =
    (k-x)\Phi\left(\frac{k-x}{\sigma}\right) + \sigma
     \phi\left(\frac{k-x}{\sigma}\right).
\end{equation*}
Note that $p$ defined by \eqref{Bput} is written 
as
\begin{equation*}
    p(x,t) = p_{\text{B}}(x,w(t)).
\end{equation*}
Then, the Bachelier UOP price function is defined as
\begin{equation*}
    p_\text{B}^{\text{UO}}(x,\sigma)
    =  p_\text{B}(x,\sigma) - p_\text{B}(2b-x,\sigma),
\end{equation*}
so that
\begin{equation*}
    u(x,t) = p_{\text{B}}^\text{UO}(x,w(t)).
\end{equation*}

\subsection{Up-and-Out Call options}
Again for $b \geq k$ with $b > X_0$,
the payoff of the Up-and-Out Call (UOC) option with barrier $b$, strike $k$ and maturity $T>0$ is
\begin{equation*}
    (X_T-k)_+ 1_{\{\tau_b > T\}}.
\end{equation*}
As in the UOP case, it suffices to find a solution 
to the boundary problem \eqref{pde} with \eqref{boundary}
and
\begin{equation*}
    u(x,T) = (x-k)_+, \ \ x \in (-\infty,b)
\end{equation*}
instead of \eqref{terminal}.
The solution is necessarily discontinuous at $(b,T)$,
which leads us to use
\begin{equation*}
    \begin{split}
        c_b(x,t) :&= \mathsf{E}[(X_T-k)_+ 1_{\{X_T < b\}}| X_t = x]\\
        &=\int_{\frac{k-x}{w(t)}}^{\frac{b-x}{w(t)}} (x + w(t) z - k)\phi(z)\mathrm{d}z \\
        &= (x-k)\left(\Phi\left(\frac{b-x}{w(t)}\right) - \Phi\left(\frac{k-x}{w(t)}\right)\right) -w(t)\left(\phi\left(\frac{b-x}{w(t)}\right) - \phi\left(\frac{k-x}{w(t)}\right)\right).
    \end{split}
\end{equation*}
Then, 
\begin{equation*}
\begin{split}
    u(x,t):&=    c_b(x,t) - c_b(2b-x,t) \\
    &=  c(x,t) - c(2b-x,t) + k-x + 2(b-k)\Phi\left(\frac{b-x}{w(t)}\right)
\end{split}
\end{equation*}
is a solution to the boundary problem,
where  $c(x,t)$ is the Bachelier call option price
\begin{equation*}
    \begin{split}
        c(x,t) & = \mathsf{E}[(X_T-k)_+|X_t = x] \\
        &= 
        \int \left(x+w(t)z - k\right)_+\phi(z)\mathrm{d}z\\
     &= \int_{\frac{k-x}{w(t)}}^\infty
     \left(x+w(t)z-k\right)\phi(z)\mathrm{d}z\\ 
     &=(x-k)\Phi\left(\frac{x-k}{w(t)}\right) + w(t)
     \phi\left(\frac{x-k}{w(t)}\right).
    \end{split}
\end{equation*}

\subsection{Other types of barrier options}

For $b\leq k$ with $b < X_0$, 
the payoff of the Down-and-Out Call (DOC) option with barrier $b$, strike $k$ and maturity $T>0$ is
\begin{equation*}
    (X_T-k)_+ 1_{\{\tau_b > T\}}.
\end{equation*}
By the same idea as for the UOP, its unique no-arbitrage price at $s < \tau_b$ is given by 
\begin{equation*}
 c(X_s,s) - c(2b-X_s,s).
\end{equation*}

For $b\leq k$ with $b < X_0$, 
the payoff of the Down-and-Out Put (DOP) option with barrier $b$, strike $k$ and maturity $T>0$ is
\begin{equation*}
    (k-X_T)_+ 1_{\{\tau_b > T\}}.
\end{equation*}
By the same idea as for the UOC, its unique no-arbitrage price at $s < \tau_b$ is given by
\begin{equation*}
 p_b(X_s,s) - p_b(2b-X_s,s),
\end{equation*}
where 
\begin{equation*}
    p_b(x,t) := \mathsf{E}[(k-X_T)_+1_{\{X_T > b\}}\mid X_t=x].
\end{equation*}

Knock-in options are replicated by vanilla options and knock-out options. 
For example, for $b\geq k$ with $b > X_0$, 
the payoff of the Up-and-In Put (UIP) option with barrier $b$, strike $k$ and maturity $T>0$ is
\begin{equation*}
    (k-X_T)_+ 1_{\{\tau_b \leq  T\}} = (k-X_T)_+ -  (k-X_T)_+ 1_{\{\tau_b >  T\}} .
\end{equation*}
Since the second term is the payoff of the UOP option with the price given by \eqref{uop}, the unique no-arbitrage price of the UIP at time $s < \tau_b$ is given by $p(2b-X_s,s)$.
Other knock-in options are also treated similarly.

\section{The convergence of the Black-Scholes price}

We next discuss the Black--Scholes model as further preparation for the asymptotic analysis in the next section.
 Consider a time-inhomogeneous Black-Scholes model
\begin{equation*}
    \frac{\mathrm{d} S_t}{S_t} = \sqrt{v(t)} \mathrm{d}W_t
\end{equation*}
for an asset price process $S$, where $W$ and $v$ are as before.
By the same argument as in the previous section,
the unique no-arbitrage price at time $0$ of the UOP option with barrier $B>S_0$, strike $K < B$ and maturity $T>0$ is given by $u(S_0,0)$, where
$u(s,t)$ is now a solution of the boundary value problem
    \begin{align}
        &\frac{\partial u}{\partial t}(s,t) + \frac{1}{2}v(t)s^2\frac{\partial^2 u}{\partial s^2}(s,t) =0,\ \  (s,t) \in (0,B) \times (-\infty,T), \label{BSpde}\\
        & u(s,T) = (K-s)_+,  \ \ s \in (0,B), \label{BSpayoff}\\
        & u(B,t) = 0, \ \  t \in (-\infty,T).
    \end{align}
Using the Black-Scholes put price function
\begin{equation*}
\begin{split}
     & u_{\text{BS}}(s,t) =s p_{\mathrm{BS}}\left(\log \frac{K}{s}, w(t)\right),\\
     & p_{\mathrm{BS}}(\ell,\sigma) = e^\ell\Phi\left(\frac{\ell}{\sigma} + \frac{\sigma}{2}\right)
     - \Phi\left(\frac{\ell}{\sigma} - \frac{\sigma}{2}\right)
\end{split}
\end{equation*}
solving \eqref{BSpde} and \eqref{BSpayoff} with $u=u_{\text{BS}}$, 
we can find a solution to the boundary value problem as
\begin{equation*}
    u(s,t) = u_{\text{BS}}(s,t)- \frac{s}{B}u_{\text{BS}}\left(\frac{B^2}{s},t\right).
\end{equation*}
We define the Black-Scholes UOP price
\begin{equation*}
    p_\text{BS}^\text{UO}(s,\sigma)
    = sp_\text{BS}\left(\log \frac{K}{s}, \sigma\right)
    - B p_\text{BS}\left(\log \frac{Ks}{B^2}, \sigma\right),
\end{equation*}
so that
\begin{equation*}
u(S_0,0)=  p_\text{BS}^\text{UO}\left( S_0,
 w(0)\right)
\end{equation*}
gives the UOP price at time $0$, where $w(t)$ is defined by \eqref{wt}. 

Now, we let
$B=S_0(1 + \sqrt{v_0\theta}b)$ with $b>0$, $K =S_0(1 + \sqrt{v_0\theta}k)$ with $k < b$, $T = \theta$ and
 $v_0 = v(0)$,
and consider the scaling limit of
the Black-Scholes UOP price
$ p_\text{BS}^\text{UO}\left( S_0, a\sqrt{\theta}\right)$
as $\theta \to 0$, where $a>0$.
A straightforward computation shows that
\begin{equation*}
     p_\text{BS}^\text{UO}\left( S_0, a\sqrt{\theta}\right) = S_0\sqrt{v_0\theta}p_\text{B}^\text{UO}\left(0, \frac{a}{\sqrt{v_0}}\right) + O(\theta),
\end{equation*}
where $p_\text{B}^\text{UO}$ is defined by \eqref{uop},
which shows that the Black-Scholes price converges to the corresponding Bachelier price.
Further, assuming $a = a(\theta) \to \sqrt{v_0}$ as $\theta \to 0$, we have
\begin{equation*}
     p_\text{BS}^\text{UO}\left( S_0, a\sqrt{\theta}\right) = S_0\sqrt{v_0\theta}p_\text{B}^\text{UO}\left(0, \frac{a}{\sqrt{v_0}}\right) 
     + S_0v_0\theta \gamma(b,k)  + 
     o(\theta),
     \end{equation*}
     where
     \begin{equation}\label{gamma}
     \gamma(b,k)=
         \frac{k}{2}\left(\phi\left(k\right)-\phi(k-2b)\right) - b(b-k)\Phi(k-2b).
     \end{equation}
Substituting
\begin{equation*}
    a = \frac{w(0)}{\sqrt{\theta}},
\end{equation*} 
we have
\begin{equation*}
  p_\text{BS}^\text{UO}\left( S_0,
  w(0)\right)   = P_0 + P_1 + o(\theta),
\end{equation*}
where $P_0$ is the Bachelier UOP price (with $\mathrm{d}X_t = S_0\sqrt{v(t)}\mathrm{d}W_t$ and $X_0 = S_0$; see \eqref{uop}) given by 
\begin{equation}\label{p0}
\begin{split}
P_0  = \, & (K -S_0)\Phi\left(\frac{K-S_0}{S_0 w(0)}\right) + S_0 w(0)
     \phi\left(\frac{K-S_0}{S_0 w(0)}\right)\\
     &- (K + S_0 -2B)\Phi\left(\frac{K+S_0-2B}{S_0 w(0)}\right) -S_0 w(0)
     \phi\left(\frac{K+S_0-2B}{S_0 w(0)}\right)
\end{split}
\end{equation}
that is of $O(\sqrt{\theta})$, and  $P_1 = O(\theta)$  given by
\begin{equation}\label{p1}
    \begin{split}
        P_1 = \, &\frac{1}{2}(K-S_0)w(0) \left(
        \phi\left(\frac{K-S_0}{S_0 w(0)}\right)
        -
          \phi\left(\frac{K+S_0-2B}{S_0 w(0)}\right)\right)\\
          & - (B-S_0)\frac{B-K}{S_0} \Phi\left(\frac{K+S_0-2B}{S_0 w(0)}\right)\\
          = \, & S_0v_0\theta \gamma(b,k)+ o(\theta).
    \end{split}
\end{equation}
    
%

\section{Short-term asymptotics}
Now we consider an abstract (local) stochastic volatility model
\begin{equation*}
    \frac{\mathrm{d}S_t}{S_t} = \sqrt{V_t}\mathrm{d}W_t,
\end{equation*}
where $V$ is a positive continuous adapted process on a filtered probability space
$(\Omega,\mathscr{F},\mathsf{P},\{\mathscr{F}_t\})$, which supports the $\{\mathscr{F}_t\}$-standard Brownian motion $W$.
We assume that $\mathscr{F}_0$ is $\mathsf{P}$-trivial, that is,
\begin{equation*}
 \mathsf{P}(A)\in\{0,1\}\qquad\text{for every }A\in\mathscr{F}_0.
\end{equation*}
Consequently, every $\mathscr{F}_0$-measurable random variable is almost surely
constant.  In particular, the adaptedness of $S$ and $V$ implies that $S_0$
and $V_0$ are deterministic and that $V_0=v_0$ almost surely.
Let $ v(t) = \mathsf{E}[V_t]$ and $v_0 = v(0)$.
We consider the pricing of the UOP option on the asset $S$ with barrier $B = S_0(1 + \sqrt{v_0\theta}b)$,
strike $K = S_0(1 + \sqrt{v_0\theta}k)$ and maturity $\theta > 0$,
where $b >\max\{k,0\}$ . The option payoff is
\begin{equation*}
    (K-S_\theta)_+ 1_{\left\{\sup_{t \in [0,\theta]}S_t \leq B\right\}}
    = \sqrt{v_0\theta}S_0(k-X^\theta_1)_+1_{\left\{\tau^\theta_b > 1\right\}},
\end{equation*}
where
\begin{equation*}
    X^\theta_u = \frac{1}{\sqrt{v_0\theta}}\left(\frac{S_{\theta u}}{S_0}-1\right), \ \ 
      \tau^\theta_b=\inf\{t \geq 0 ; X^\theta_t = b\}.
\end{equation*}
An arbitrage-free option price at time $0$ is given by
\begin{equation*}
p^\text{UO}:=
    \sqrt{v_0\theta}S_0 
    \mathsf{E}\left[(k-X^\theta_1)_+1_{\left\{\tau^\theta_b > 1\right\}}\right].
\end{equation*}
Let 
\begin{equation*}
    Y^\theta_u = \frac{1}{\theta^H}\left(\frac{V_{\theta u}}{v(\theta u)}-1\right)
\end{equation*}
for $H \in (0,1/2]$.

\begin{theorem}\label{thm1}
    Assume that  $v(t)$ is continuous at $t=0$, $Y^{\theta}_1$ is uniformly integrable, and 
\begin{equation}\label{conv}
    \left(X^\theta_1,Y^\theta_1, \sup_{u \in [0,1]}X^\theta_u\right)
    \to (\xi,\eta,\zeta)
\end{equation}
in law as $\theta \to 0$. Then,
\begin{equation}\label{main}
\begin{split}
   p^\text{UO} = 
         p_\text{BS}^\text{UO}\left( S_0,
  \sqrt{\int_0^\theta v(r)\mathrm{d}r}\right) 
         + \alpha \sqrt{v_0}S_0 \theta^{H+1/2}+ o(\theta^{H+1/2}),
\end{split} 
\end{equation}
where
\begin{equation*}
    \alpha =     \frac{1}{2}\int_0^1 u^H \int_{-\infty}^b 
        \mathsf{E}\left[\eta \,\Big| \,\zeta \leq \frac{b}{\sqrt{u}}, \xi =\frac{x}{\sqrt{u}} \right] \psi(b-x,k-x,1-u) 
         \psi\left(b,x,u\right)\, \mathrm{d}x \mathrm{d}u
\end{equation*}
and
\begin{equation*}
\psi(m,x,t) = \phi(x,t)-\phi(2m-x,t), \ \ 
    \phi(x,t) = \frac{1}{\sqrt{2\pi t}} \exp\left(-\frac{x^2}{2t} \right).
\end{equation*}
\end{theorem}
\begin{remark} \upshape
    Fukasawa~\cite{F21} derives an asymptotic expansion of the put option price and the implied volatility under
    the assumption $(X^\theta_1,Y^\theta_1) \to (\xi,\eta)$.
    We extend this approach to UOP options by additionally assuming the joint convergence of the running maximum $\sup_u X^\theta_u$.
\end{remark}

{\it Proof:}
Note that
$X^\theta_u = \sqrt{u}X^{\theta u}_1$, $Y^\theta_u = u^H Y^{\theta u}_1$
and
\begin{equation*}
    \sup_{w \in [0,u]} X^\theta_w = \sqrt{u}\sup_{w \in [0,1]} X^{\theta u}_w.
\end{equation*}
Therefore, under \eqref{conv}, we have
\begin{equation}\label{conv2}
    \left(X^\theta_u,Y^\theta_u, \sup_{w \in [0,u]}X^\theta_w\right)
    \to (\sqrt{u}\xi,u^H\eta,\sqrt{u}\zeta)
\end{equation}
in law for each $u \in [0,1]$ as $\theta \to 0$.
Note also that
\begin{equation}\label{qv}
    \mathrm{d}\langle X^\theta \rangle_u = \frac{1}{v_0}\left(\frac{S_{\theta u}}{S_0}\right)^2 V_{\theta u}\mathrm{d}u = \frac{1}{v_0}(1 + \sqrt{v_0\theta}X^\theta_u)^2 V_{\theta u}\mathrm{d}u
\end{equation}
and so,
\begin{equation*}
    \langle X^\theta \rangle_u \to u
\end{equation*}
in probability.
Hence, the martingale central limit theorem implies that $X^\theta$
  converges weakly to a standard Brownian motion.
 In particular, 
 the law of $(\xi, \zeta)$ is identified as that of
\begin{equation*}
   \left(W_1, \sup_{t \in [0,1]}W_t \right).
\end{equation*}
Let 
\begin{equation}
    p^\theta_b(x,t) = p(x,t) - p(2b-x,t),
\end{equation}
where $p(x,t)$ is defined as \eqref{Bput} with $T=1$ and $w$ replaced by
\begin{equation*}
    w^\theta(t) := \sqrt{\frac{1}{v_0}\int_t^1 v(\theta u) \mathrm{d}u}.
\end{equation*}
Note that
\begin{equation*}
    p^\theta_b(0,0) =p_\text{B}^\text{UO}\left(0, \sqrt{\frac{1}{v_0 \theta}\int_0^\theta v(r)\mathrm{d}r}\right)
\end{equation*}
and
\begin{equation*}
    S_0\sqrt{v_0\theta}p^\theta_b(0,0) = P_0,
\end{equation*}
where $P_0$ is defined by \eqref{p0}.
Since $p^\theta_b$ is a continuous solution of
\begin{equation}\label{pde2}
    \begin{split}
         &\frac{\partial p^\theta_b}{\partial t}(x,t) + \frac{v(\theta t)}{2v_0}\frac{\partial^2 p^\theta_b}{\partial x^2}(x,t) =0,\ \  (x,t) \in (-\infty,b) \times (-\infty,1), \\
        & p^\theta_b(x,1) = (k-x)_+, \ \ x \in (-\infty,b), \\
        & p^\theta_b(b,t) = 0, \ \  t \in (-\infty,1),
    \end{split}
\end{equation}
the payoff is expressed as
\begin{equation*}
    (k-X^\theta_1)_+1_{\left\{\tau^\theta_b > 1\right\}}
    =  p^\theta_b(X^\theta_{1 \wedge \tau^\theta_b}, 1 \wedge \tau^\theta_b)
\end{equation*}
and by It\^o's formula and \eqref{qv}, we have
\begin{equation*}\begin{split}
    p^\theta_b(X^\theta_{\tau^\theta_b \wedge 1},&\tau^\theta_b\wedge 1)  \\ = 
    p^\theta_b(0,0)  &+ \int_0^{\tau^\theta_b \wedge 1}\frac{\partial p^\theta_b}{\partial x}(X^\theta_u,u) \mathrm{d}X^\theta_u \\
    & -\int_0^{\tau^\theta_b \wedge 1}
    \frac{\partial p^\theta_b}{\partial t}(X^\theta_u,u)\left((1 + \sqrt{v_0\theta}X^\theta_u)^2 \frac{V_{\theta u}}{v(\theta u)}-1\right)\mathrm{d}u. 
\end{split}
\end{equation*}
Since $X^\theta_u \geq -(v_0\theta)^{-1/2}$ and $p^\theta_b(x,s)$ is bounded on $[-(v_0\theta)^{-1/2},b]\times [0,1]$,
the bounded convergence theorem with a localization argument yields
\begin{equation*}
\begin{split}
    &  \mathsf{E}[ p^\theta_b(X^\theta_{\tau_b^\theta \wedge 1},\tau_b^\theta\wedge 1) ] \\ & = 
   p^\theta_b(0,0)- \mathsf{E}[\int_0^{\tau^\theta_b \wedge 1}
    \frac{\partial p^\theta_b}{\partial t}(X^\theta_u,u)\left((1 + \sqrt{v_0\theta}X^\theta_u)^2 \frac{V_{\theta u}}{v(\theta u)}-1\right)\mathrm{d}u].
\end{split}  
   \end{equation*}
We justify the passage to the limit as in Step~1 of the proof of
Fukasawa~\cite{F21}.  First, the continuity and positivity of $v$ at zero
imply that, for all sufficiently small $\theta$, there are constants
$0<c<C<\infty$ such that
\begin{equation*}
 c\sqrt{1-u}\leq w^\theta(u)\leq C\sqrt{1-u},\qquad 0\leq u<1.
\end{equation*}
Moreover,
\begin{equation*}
    \frac{\partial p^\theta_b}{\partial t}(x,t) =
    \frac{v(\theta t)}{2v_0w^\theta(t)} \left(
    \phi\left(\frac{k+x-2b}{w^\theta(t)}\right)
    - \phi\left(\frac{k-x}{w^\theta(t)}\right) \right).
\end{equation*}
Consequently, using the Gaussian decay, for $j=0,1,2$ there is a constant
$C_j$ independent of $\theta$ such that
\begin{equation}\label{domination}
 \sup_{x<b}|x|^j\left|\frac{\partial p^\theta_b}{\partial t}(x,u)\right|
 \leq \frac{C_j}{\sqrt{1-u}},\qquad 0\leq u<1.
\end{equation}
This is the barrier analogue of the bound used in~\cite{F21}.

By the scaling relation $Y^\theta_u=u^H Y^{\theta u}_1$ and the uniform
integrability of $\{Y^\delta_1\}_{\delta>0}$,
\begin{equation*}
 \sup_{0<\theta\leq\theta_0}\mathsf{E}[|Y^\theta_u|]\leq C u^H,
 \qquad 0<u\leq1,
\end{equation*}
Thus the terms containing $Y^\theta_u$ are dominated in expectation by a
constant times $u^H(1-u)^{-1/2}$, whereas the remaining terms are dominated
by a constant times $(1-u)^{-1/2}$.  Both bounds are integrable over
$(0,1)$.  For each fixed $u\in(0,1)$, \eqref{conv2}, uniform integrability,
and $\mathsf{P}(\sqrt{u}\zeta=b)=0$ give convergence of the corresponding
expectations.  Finally, expanding
\begin{equation*}
 (1+\sqrt{v_0\theta}X^\theta_u)^2(1+\theta^H Y^\theta_u)-1
\end{equation*}
and using \eqref{domination} shows that the terms containing
\begin{equation*}
 \theta(X^\theta_u)^2,\qquad
 \theta^{H+1/2}X^\theta_uY^\theta_u,\qquad
 \theta^{H+1}(X^\theta_u)^2Y^\theta_u
\end{equation*}
are $o(\theta^H)$ after integration.
The dominated convergence theorem therefore yields
   \begin{equation*}
       \begin{split}
            \mathsf{E}[&\int_0^{\tau^\theta_b \wedge 1}
    \frac{\partial p^\theta_b}{\partial t}(X^\theta_u,u)\left((1 + \sqrt{v_0\theta}X^\theta_u)^2 \frac{V_{\theta u}}{v(\theta u)}-1\right)\mathrm{d}u]\\
    = & 2\sqrt{v_0\theta}\int_0^1  \mathsf{E}\left[
     \frac{\partial p^0_b}{\partial t}(\sqrt{u}\xi,u)\sqrt{u}\xi  1_{\{\sqrt{u}\zeta \leq b\}}  \right]\mathrm{d}u
    \\ &+\theta^{H}\int_0^1  \mathsf{E}\left[
     \frac{\partial p^0_b}{\partial t}(\sqrt{u}\xi,u) u^H \eta 1_{\{\sqrt{u}\zeta \leq b\}} 
    \right] \mathrm{d}u + o(\theta^H).
       \end{split}
   \end{equation*}
The time derivative converges to
    \begin{equation*}
       \frac{\partial p^0_b}{\partial t}(x,t) =
       \frac{1}{2\sqrt{1-t}} \left(
\phi\left(\frac{k+x-2b}{\sqrt{1-t}}\right)
- \phi\left(\frac{k-x}{\sqrt{1-t}}\right) \right) = -\frac{1}{2}\psi(b-x,k-x,1-t)
   \end{equation*}
   uniformly on compact sets. Therefore, using \eqref{RP} and \eqref{TC},
we have
\begin{equation*}
    \begin{split}
        &\mathsf{E}\left[
     \frac{\partial p^0_b}{\partial t}(\sqrt{u}\xi,u)\sqrt{u}\xi  1_{\{\sqrt{u}\zeta \leq b\}}  \right]\\ &= 
     -\frac{\sqrt{u}}{2}\int_{-\infty}^{b/\sqrt{u}} x
     \psi(b-\sqrt{u}x,k-\sqrt{u}x,1-u)
    \psi\left(\frac{b}{\sqrt{u}},x,1\right)\, \mathrm{d}x
     \\ &= 
     -\frac{1}{2}\int_{-\infty}^b x\psi(b-x,k-x,1-u)\psi\left(b,x,u\right)\, \mathrm{d}x
     \\ &= -\frac{1}{2}\bar{\psi}(b,k,u),
    \end{split}
\end{equation*}
where
\begin{equation*}
\begin{split}
&\bar{\psi}(m,x,t)=    \bar{\phi}(m,x,t)-\bar{\phi}(m,2m-x,t), \\
      &\bar{\phi}(m,x,t) = \phi(x)\left(2m-tx-2(m-tx)\Phi\left(\frac{m-tx}{\sqrt{t(1-t)}}\right)\right).
\end{split}
\end{equation*}
Similarly,
\begin{equation*}
    \begin{split}
       &  \mathsf{E}\left[
     \frac{\partial p^0_b}{\partial t}(\sqrt{u}\xi,u) u^H \eta 1_{\{\sqrt{u}\zeta \leq b\}} 
    \right] \\
    &= -\frac{u^H}{2}\int_{-\infty}^{b/\sqrt{u}}
     \mathsf{E}\left[\eta \,\Big| \,\zeta \leq \frac{b}{\sqrt{u}}, \xi =x  \right]
      \psi(b-\sqrt{u}x,k-\sqrt{u}x,1-u)
    \psi\left(\frac{b}{\sqrt{u}},x,1\right)\, \mathrm{d}x\\
    &= -\frac{u^H}{2}\int_{-\infty}^b
     \mathsf{E}\left[\eta \,\Big| \,\zeta \leq \frac{b}{\sqrt{u}}, \xi =\frac{x}{\sqrt{u}}  \right]
      \psi(b-x,k-x,1-u)
    \psi\left(b,x,u\right)\, \mathrm{d}x.
    \end{split}
\end{equation*}
It remains to show
\begin{equation}\label{gamma2}
    \int_0^1\bar{\psi}(b,k,u) \mathrm{d}u = \gamma(b,k).
\end{equation}
To verify this identity, compute the same Black--Scholes barrier price in two
ways.  When $V_t=v(t)$ is deterministic, $Y^\theta\equiv0$, and the stopped
It\^o calculation above, before using \eqref{gamma2}, gives
\begin{equation*}
 p^{\mathrm{UO}}
 =P_0+S_0v_0\theta\int_0^1\bar\psi(b,k,u)\,\mathrm{d}u+o(\theta).
\end{equation*}
On the other hand, the direct expansion of the closed-form Black--Scholes
barrier price in the previous section gives
$p^{\mathrm{UO}}=P_0+S_0v_0\theta\gamma(b,k)+o(\theta)$; see
\eqref{p0}--\eqref{p1}.  Comparing the two independently computed
coefficients proves \eqref{gamma2}.  This comparison is simply an alternative
calculation of the same Black--Scholes price and does not use the conclusion
\eqref{main}.
 \hfill{$\square$}

\begin{remark}\upshape
The equality \eqref{gamma2} holds only for $k \leq b$.
Indeed, it is easy to see $\bar{\psi}(b,k,u) = -\bar{\psi}(b,2b-k,u)$ while \eqref{gamma} does not have this skew symmetry. The discrepancy for $k>b$ can also be observed numerically; see Figure~\ref{fig:gamma}.
\end{remark}

\begin{figure}[htbp]
    \centering
     \includegraphics[width=9cm]{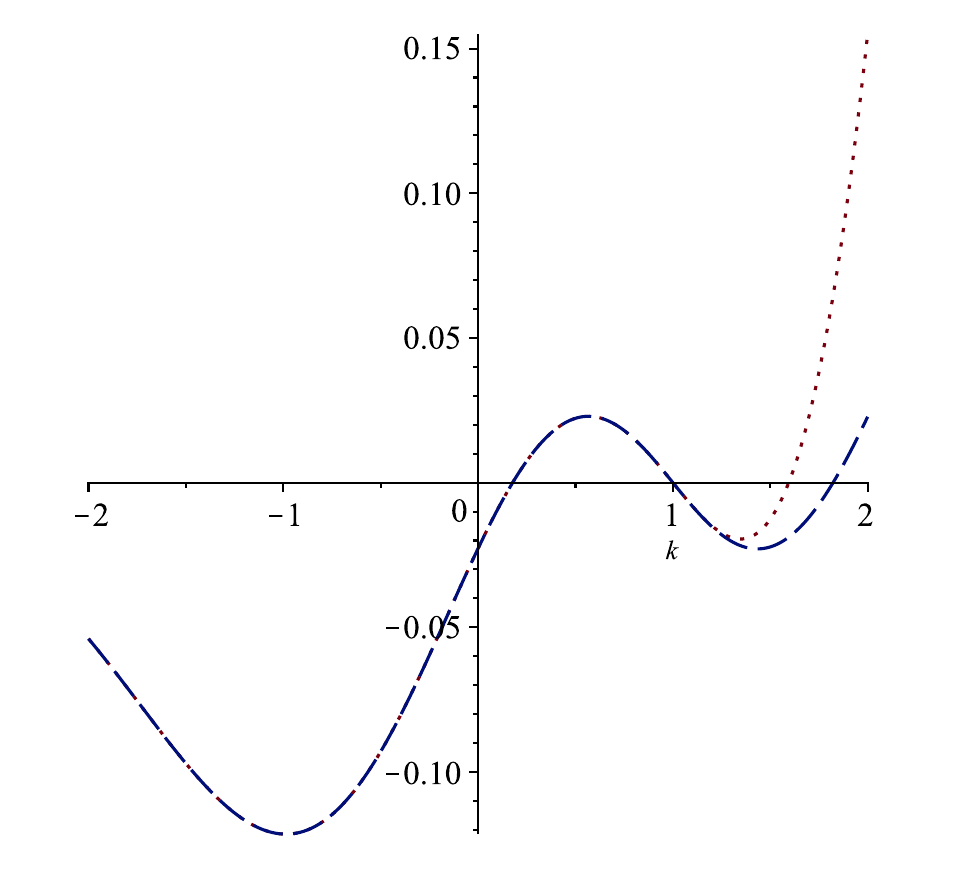}
    \caption{The dotted curve is $k \mapsto \gamma(1,k)$
    defined by  \eqref{gamma}. The dashed curve is by a numerical integration of $\bar{\psi}(1,k,u)$ with respect to $u$ over $[0,1]$.}
    \label{fig:gamma}
\end{figure}

\begin{corollary}\label{cor1}
Under the assumptions of Theorem~\ref{thm1} with $H=1/2$, if 
$\mathsf{E}[|\eta\xi|] < \infty$ and
\begin{equation*}
 \eta^\prime:=   \eta -\mathsf{E}[\eta \xi] \xi
\end{equation*}
is independent of $(\xi,\zeta)$, then,
\begin{equation*}
\begin{split}
    p^\text{UO} &= 
         p_\text{BS}^\text{UO}\left( S_0,
  w(0)\right) 
         + \frac{\mathsf{E}[\eta \xi]}{2}\gamma(b,k)\sqrt{v_0}S_0 \theta+ o(\theta),\\
         & =
         p_\text{BS}^\text{UO}\left( S_0,
  w(0)\right) 
         + \frac{\mathsf{E}[\eta \xi]}{2\sqrt{v_0}}
         \left( p_\text{BS}^\text{UO}\left( S_0,
  w(0)\right) - P_0\right) + o(\theta),
\end{split}
\end{equation*}
where $P_0$ denotes the Bachelier price defined in \eqref{p0}.
\end{corollary}

\begin{remark}
    \upshape When the strike price $K$ is above the barrier level $B$, the UOP price admits a model-free formula. Assume the underlying asset $S$ to be a continuous martingale and let $\tau_B$ denote the first hitting time to $B > S_0$. By the optional sampling theorem, we have
    \begin{equation*}
        S_0 = \mathsf{E}[S_{\tau_B \wedge T}] = 
        \mathsf{E}[S_T1_{\{\tau_B > T\}}] + B \mathsf{P}(\tau_B \leq T).
    \end{equation*}
    On the other hand, when $K \geq B$, $K-S_T$ cannot be negative unless hitting the barrier before $T$ and so $(K-S_T)_+1_{\{\tau_B > T\}} = (K-S_T)1_{\{\tau_B > T\}}$,
    which in turn implies
    \begin{equation*}
        \mathsf{E}[(K-S_T)_+1_{\{\tau_B > T\}}]
= K \mathsf{P}(\tau_B>T) -\mathsf{E}[S_T 1_{\{\tau_B > T\}}]
= K-S_0 + (B-K)\mathsf{P}(\tau_B \leq T).
    \end{equation*}
    In particular, we have $\mathsf{E}[(B-S_T)_+1_{\{\tau_B > T\}}] = B-S_0$, which
    we have used for variance reduction of the Monte Carlo presented in the next section.
\end{remark}

\begin{remark}\upshape
The UOC and DOP cases contain an additional regularity difficulty that
does not arise for the UOP.  For the UOC with $b>k$ and the DOP with $b<k$,
the terminal payoff is incompatible with the absorbing boundary condition:
\begin{equation*}
 \lim_{x\uparrow b}(x-k)_+=b-k>0,\qquad
 \lim_{x\downarrow b}(k-x)_+=k-b>0,
\end{equation*}
respectively.  Equivalently, their reflection formulas contain a digital
barrier-survival component.  The time derivative of this component develops a
boundary layer of order $(1-u)^{-1}$ when the distance to the barrier is of
order $\sqrt{1-u}$.  Consequently, the integrable uniform estimate
\eqref{domination} used in the dominated-convergence argument for
Theorem~\ref{thm1} is not directly available.  Weak convergence and uniform
integrability alone therefore do not justify the analogous limit.  This
obstruction disappears in the compatible case $b=k$.

We conjecture that the expansion \eqref{main} nevertheless has natural UOC and
DOP analogues, with the corresponding Black--Scholes barrier price as the
leading term and with the first correction obtained from the appropriate
killed Bachelier pricing function and running maximum or minimum.  Such an
extension should hold under an additional quantitative nondegeneracy
condition that yields smooth killed transition densities, together
with uniform Gaussian bounds on the densities and their first derivatives near
the barrier.  These bounds should permit the digital boundary layer to be
integrated before passage to the limit, replacing the pointwise domination
used in the UOP proof.
\end{remark}

\section{Examples}
\subsection{Regular local stochastic volatility models}
For the regular models below, suppose that the variance is an It\^o
semimartingale $V_t=V_0+M_t+\int_0^t a_s\,\mathrm{d}s$, where $V_0=v_0>0$
is deterministic, $M$ is a true martingale with $M_0=0$, and
$\sup_{s\leq t_0}\mathsf{E}[|a_s|]<\infty$ for some $t_0>0$.
Then
\begin{equation*}
 v(t)=\mathsf{E}[V_t]=v_0+O(t),\qquad
 w(0)^2=\int_0^\theta v(t)\,\mathrm{d}t
 =v_0\theta+O(\theta^2),
\end{equation*}
and hence $w(0)=\sqrt{v_0\theta}+O(\theta^{3/2})$.
The Black--Scholes and Bachelier UOP prices have bounded derivatives
with respect to total volatility for the strike and barrier scaling
considered here.  Freezing the variance at its initial value therefore
changes these prices by $O(\theta^{3/2})=o(\theta)$ and preserves the
first-order regular-model expansion.  We use this initial-variance
convention in the numerical experiments below.  

The maturity-averaged forward variance can also be replaced by the
at-the-money implied variance.  Indeed, under the regularity and
Gaussian-limit assumptions of Fukasawa~\cite[Corollary~2.1 and
Remark~2.4]{F21}, the leading implied-volatility correction vanishes
at the money.  If $\sigma_{\mathrm{ATM}}(\theta)$ denotes the
Black--Scholes implied volatility of a vanilla option with strike
$S_0$ and maturity $\theta$, then, in the regular case $H=1/2$,
\begin{equation*}
 \sigma_{\mathrm{ATM}}(\theta)^2
 =\frac{1}{\theta}\int_0^\theta v(t)\,\mathrm{d}t
 +o(\theta^{1/2}),\qquad
 \sigma_{\mathrm{ATM}}(\theta)\sqrt{\theta}=w(0)+o(\theta).
\end{equation*}
Consequently, using ATM implied variance in place of the averaged
forward variance changes the Black--Scholes and Bachelier UOP prices
by $o(\theta)$ and preserves the first-order expansion.

Consider an abstract local stochastic volatility model
\begin{equation}\label{LSV}
    \frac{\mathrm{d}S_t}{S_t} = \sigma(S_t,Y_t,t)\mathrm{d}W_t, \ \ 
    \mathrm{d}Y^i_t = b^i(S_t,Y_t,t)\mathrm{d}t + \sum_{j=1}^dc^i_j(S_t,Y_t,t)\mathrm{d}B^j_t,
    \end{equation}
    where $(B^1,\dots,B^d)$ is a $d$-dimensional Brownian motion with
    \begin{equation*}
        \mathrm{d}\langle B^j, W\rangle_t = \rho^j \mathrm{d}t, \ \ \sum_{j=1}^d|\rho^j| \leq 1.
    \end{equation*}
    If $\sigma$, $b^i$ and $c^i_j$ are sufficiently regular,  by It\^o's formula, 
    $v(t) = V_0 + O(t)$, and
    \begin{equation*}
        \frac{1}{\sqrt{\theta}}(V_\theta -V_0) = S_0 \frac{\partial \sigma^2}{\partial s}(S_0,Y_0,0) \frac{1}{\sqrt{\theta}}\left(\frac{S_\theta}{S_0} -1\right)
        + \sum_{i,j}  \frac{\partial \sigma^2}{\partial y^i} c^i_j(S_0,Y_0,0) \frac{B^j_\theta}{\sqrt{\theta}} + o_p(1).
    \end{equation*}
    Therefore, for $H=1/2$,
    \begin{equation*}
        \frac{1}{\theta^H}
      \left(\frac{V_\theta}{v(\theta)} -1\right)
      \to \left( \frac{S_0}{\sqrt{v_0}}\frac{\partial \sigma^2}{\partial s}(S_0,Y_0,0) + \sum_{i,j}\rho^j c^i_j \frac{\partial}{\partial y^i} \log \sigma^2 (S_0,Y_0,0)\right)\xi + \eta^\prime
    \end{equation*}
    in law, where $\eta^\prime$ is independent of $(\xi,\zeta)$.
    Corollary~\ref{cor1} applies with
    \begin{equation*}
    \frac{\mathsf{E}[\eta \xi]}{2} =     S_0\frac{\partial \sigma}{\partial s}(S_0,Y_0,0) + \sum_{i,j}\rho^j c^i_j \frac{\partial}{\partial y^i} \log \sigma (S_0,Y_0,0)
    \sim 2 A,
    \end{equation*}
    where $A$ denotes the Black-Scholes at-the-money volatility skew under the model; see Fukasawa~\cite{F17}.
    Thus, up to an $o(\theta)$ error, the UOP price is determined by European option prices in a model-independent manner for regular models.

    Consequently, the UOP price is approximated by
    \begin{equation*}
      p^{\text{UO}} \approx 
      p_\text{BS}^\text{UO}  + \frac{2A}{\sqrt{v_0}} (
   p_\text{BS}^\text{UO} -P_0) \approx
      P_0 + \left(1 + \frac{2A}{\sqrt{v_0}}\right)P_1 
    \end{equation*}
up to $o(\theta)$, where $p_\text{BS}^\text{UO} = p_\text{BS}^\text{UO}\left( S_0,
  w(0)\right)$ is the Black-Scholes price
  and $P_0$ is the Bachelier price defined by \eqref{p0}.

When $0 \geq 2A\geq -\sqrt{v_0}$, the approximated price is an interpolation (weighted average) of the Black-Scholes and Bachelier prices. In particular, it is free from  butterfly arbitrage opportunities.
For example, consider the CEV model
\begin{equation*}
    \mathrm{d}S_t = \sigma S_t^\beta \mathrm{d}W_t,
\end{equation*}
where $\beta \in [0,1]$.
Although the model fails to satisfy the positivity assumption, a formal application of the above formula provides an approximate price
\begin{equation*}
     p^{\text{UO}} \approx 
      P_0  +  (
   p_\text{BS}^\text{UO} -P_0)\beta.
\end{equation*}
This approximation is natural because the CEV diffusion coefficient
interpolates between the Bachelier ($\beta=0$) and Black--Scholes
($\beta=1$) coefficients.  For the numerical experiment we take
$\beta=0.5$, $\sigma=2$, $\theta=1$, $S_0=100$, $B=120$, and
21 equally spaced strikes $K\in[80,120]$.  We retain the practical
initial-volatility proxy $w(0)=\sqrt{v_0\theta}=0.2$ in both
$P_0$ and $p_{\mathrm{BS}}^{\mathrm{UO}}$.  This is a formal CEV
approximation, rather than an application of the theorem with an exact
forward variance curve; the CEV process is absorbed at zero and falls
outside its strict positivity assumptions.

The three numerical experiments use $200{,}000$ production paths and
$25{,}000$ equally spaced time steps over the respective maturity,
with $\Delta t=\theta/25{,}000$.  Barrier crossing is checked only at
the grid endpoints.  For CEV we use the absorbed Euler scheme
\begin{equation*}
 S_{i+1}=\max\{S_i+\sigma S_i^\beta\Delta W_i,0\},
\end{equation*}
with zero absorbing.  For lognormal SABR we use
\begin{equation*}
 S_{i+1}=S_i(1+\sigma_i\Delta W_i),\qquad
 \sigma_{i+1}=\sigma_i
 \exp\{\nu\Delta B_i-\nu^2\Delta t/2\},
\end{equation*}
where $\mathsf{E}[\Delta W_i\Delta B_i]=\rho\Delta t$.
Thus the volatility is sampled exactly on the grid and the stock uses
its left-endpoint value.  The rough Bergomi discretization is described below.

For each strike, let
\begin{equation*}
 Y_K=(K-S_N)^+\mathbf{1}_{\{\max_{0\leq i\leq N}S_i<B\}},
 \qquad N=25{,}000.
\end{equation*}
We use the fitted control estimator
\begin{equation*}
 \widehat p_K=\overline Y_K-\widehat\beta_K
 (\overline Y_B-(B-S_0)),\qquad
 \widehat\beta_K=
 \frac{\widehat{\operatorname{Cov}}(Y_K,Y_B)}
 {\widehat{\operatorname{Var}}(Y_B)}.
\end{equation*}
The coefficients are estimated on $10{,}000$ independent pilot paths
and then held fixed for the production sample, at every strike.
The centering value $B-S_0=20$ comes from the continuously monitored
martingale identity.  It need not equal the mean of the discretely
monitored control, so the estimator retains a finite-grid bias.
The figures show pointwise 95\% Monte Carlo confidence bands and signed
approximation errors.  Each band is the Monte Carlo estimate plus or minus
$1.96$ estimated standard errors, giving an approximate pointwise 95\%
confidence interval conditional on the fitted pilot coefficients.  Reported
standard errors are unscaled; $1.96$ times the standard error is the 95\%
confidence half-width.  The bands measure sampling uncertainty only;
they exclude monitoring and process-discretization error.  At $K=B$
the estimate is identically 20 and the band vanishes by construction.

To assess the monitoring effect, we also check the barrier at every
second endpoint on the same simulated stock paths, giving $12{,}500$
monitoring dates, and fit the corresponding control on the pilot sample.
We report paired differences from the $25{,}000$-date estimates.
This comparison changes only the monitoring frequency; it does not
check convergence of the stock or volatility discretization.

We also superpose the explicit $\gamma$ correction
\begin{equation}\label{regular-gamma-approx}
 \widehat p_\gamma=p_{\mathrm{BS}}^{\mathrm{UO}}
 +cS_0v_0\theta\,\gamma(b,k),\qquad
 b=\frac{B-S_0}{S_0\sqrt{v_0\theta}},\quad
 k=\frac{K-S_0}{S_0\sqrt{v_0\theta}},
\end{equation}
where $c=\beta-1$ for CEV and $c=\nu\rho/\sigma_0$ for lognormal
SABR.  Here $\gamma$ is given by \eqref{gamma} and is normalized using
the initial variance $v_0$, while the Black--Scholes baseline retains
the total volatility specified for each experiment.  This is the explicit first correction in Corollary~\ref{cor1}.
In the figures, ``Difference correction'' denotes
$p_{\mathrm{BS}}^{\mathrm{UO}}+c(p_{\mathrm{BS}}^{\mathrm{UO}}-P_0)$,
and ``$\gamma$ correction'' denotes \eqref{regular-gamma-approx}.

Figure~\ref{fig:cev} compares the two CEV corrections with the
endpoint-monitored benchmark.  Across the 21 strikes, the mean absolute
error decreases from $0.09426$ for Black--Scholes to
$0.01248$ for the difference correction, and the maximum
absolute error decreases from $0.22547$ to $0.01747$.
The $\gamma$ correction has mean and maximum absolute errors
$0.00674$ and $0.01983$, respectively.
Increasing the monitoring frequency from $12{,}500$ to $25{,}000$
dates on the same stock paths changes the controlled estimates by at
most $0.00477$; the largest standard error of a paired
change is $0.00041$.

As an independent continuously monitored benchmark, we solve the CEV
pricing PDE on $[0,B]$ with $u(0,t)=K$ and $u(B,t)=0$, using
Crank--Nicolson with four initial backward-Euler half steps.
Refinement from 1200 spatial cells and 2048 time steps to 2400 cells
and 4096 steps changes prices by less than $0.00002$.
The refined PDE prices differ from the Monte Carlo estimates by at
most $0.01725$, or $2.60$ Monte Carlo standard errors.
Against the PDE benchmark, the difference correction has mean and
maximum absolute errors $0.00429$ and $0.01445$.
These comparisons distinguish the approximation error from sampling
uncertainty and the residual bias of the endpoint-monitored estimator.
Both corrections improve the aggregate errors substantially, but their
small Monte Carlo discrepancies should not be interpreted as exact
errors against the continuously monitored price.
\begin{figure}[htbp]
 \centering
 \includegraphics[width=0.98\linewidth]{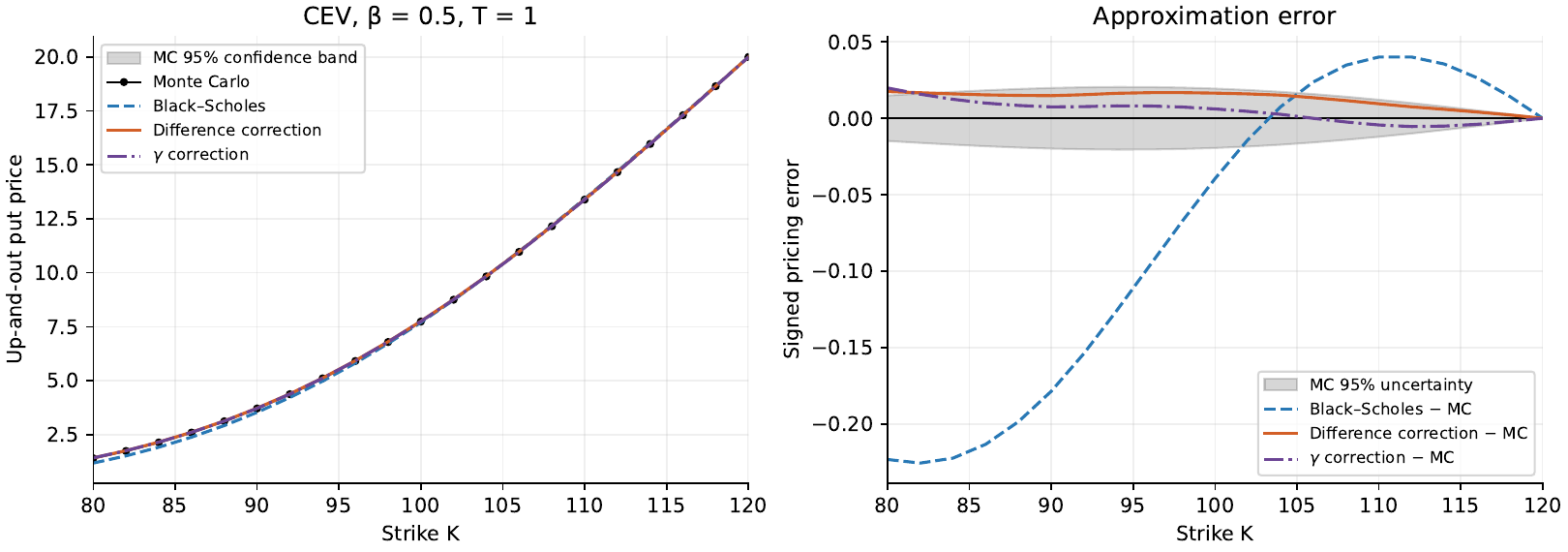}
 \caption{CEV up-and-out put prices with $\beta=0.5$, $\sigma=2$,
 $S_0=100$, $B=120$, and $\theta=1$.  Shaded regions are pointwise
 95\% Monte Carlo confidence bands; the right panel shows each
 approximation minus Monte Carlo, including the explicit $\gamma$
 correction \eqref{regular-gamma-approx}.}
 \label{fig:cev}
\end{figure}

\begin{samepage}
Under the lognormal SABR model
\begin{equation*}
    \mathrm{d}S_t = \sigma_t S_t\mathrm{d}W_t, \ \ 
    \mathrm{d}\sigma_t = \nu \sigma_t \mathrm{d}B_t,\ \ \mathrm{d}\langle W,B\rangle_t = \rho \mathrm{d}t,
\end{equation*}
the formula implies
\begin{equation*}
     p^{\text{UO}} \approx 
      p_\text{BS}^\text{UO} +  (
   p_\text{BS}^\text{UO} -P_0)\frac{\nu \rho}{\sigma_0}.
\end{equation*}
\end{samepage}
We use $\nu=0.2$, $\rho=-0.9$, $\sigma_0=0.2$, $\theta=1$,
$S_0=100$, $B=120$, and the same 21 strikes and Monte Carlo protocol
as for CEV.  Both $p_{\mathrm{BS}}^{\mathrm{UO}}$ and $P_0$ use the
initial-variance total volatility $w(0)=\sigma_0\sqrt{\theta}=0.2$.
Indeed, $V_t=\sigma_t^2$ satisfies
$\mathrm{d}V_t=\nu^2V_t\,\mathrm{d}t+2\nu V_t\,\mathrm{d}B_t$, so
$\mathsf{E}[V_t]=\sigma_0^2 e^{\nu^2t}=v_0+O(t)$, verifying the
short-time justification above.  Replacing the full difference
$p_{\mathrm{BS}}^{\mathrm{UO}}-P_0$ by $S_0v_0\theta\gamma(b,k)$
gives the superposed, asymptotically equivalent approximation
\eqref{regular-gamma-approx}; the two can differ numerically at this
one-year maturity.

Figure~\ref{fig:sabr} shows that the SABR correction also improves
the aggregate error.  The mean absolute error decreases from
$0.14937$ for Black--Scholes to $0.07444$ for the
difference correction, while the maximum absolute error decreases
from $0.34492$ to $0.09995$.
The $\gamma$ correction has mean and maximum absolute errors
$0.06009$ and $0.09242$, respectively.
The improvement is not uniform in strike: Black--Scholes is closer to
Monte Carlo at $K=98,100,102,104$.
The largest Monte Carlo standard error is $0.01079$.
Increasing the monitoring frequency from $12{,}500$ to $25{,}000$
dates changes the controlled estimates by at most
$0.00408$, with a maximum paired standard error
of $0.00039$.
The monitoring effect is detectable despite the fine grid; these
figures do not provide a rigorous bound on the remaining bias.

\begin{figure}[htbp]
    \centering
    \includegraphics[width=0.98\linewidth]{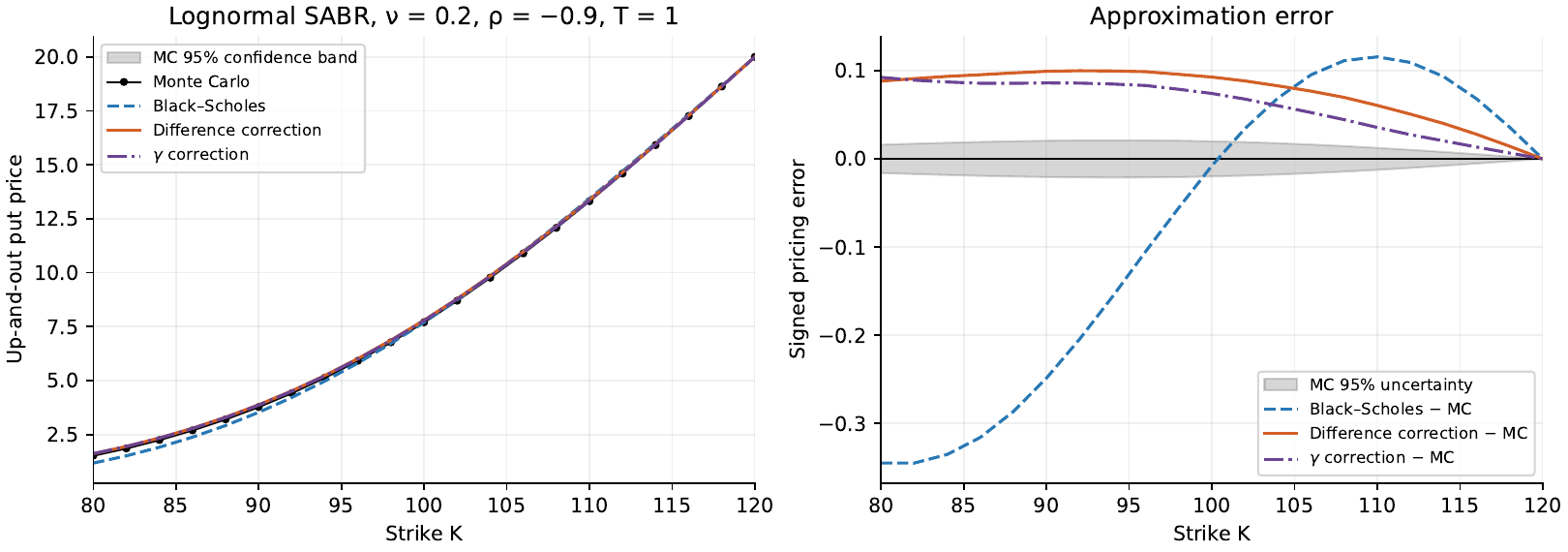}
    \caption{Lognormal SABR up-and-out put prices with $\nu=0.2$,
    $\rho=-0.9$, $\sigma_0=0.2$, and $\theta=1$.  The spot, barrier,
    strikes, and plot conventions are as in Figure~\ref{fig:cev}.
    All Black--Scholes and Bachelier terms use the initial variance
    $v_0=\sigma_0^2=0.04$, as does the normalization of $\gamma$.}
    \label{fig:sabr}
\end{figure}

\begin{samepage}
\begin{remark} \upshape
    With the $Y$ process frozen in \eqref{LSV},
    the UOP price under the corresponding local volatility model is approximated 
    as
    \begin{equation*}
        p^\text{UO}_\text{loc} \approx P_0 + \left(1 + \frac{S_0}{\sigma} \frac{\partial \sigma}{\partial s}(S_0,Y_0,0)\right)(p^\text{UO}_\text{BS}-P_0).
    \end{equation*}
    We can use this price as the leading term of the expansion:
    \begin{equation*}
        p^\text{UO} \approx p^\text{UO}_\text{loc} + (P_0 - p^\text{UO}_\text{loc})\delta, \ \ \delta = - \frac{\sum_{i,j}\rho^jc^i_j\frac{\partial}{\partial y^i}\log \sigma }{\sigma + S_0 \frac{\partial \sigma}{\partial s} }(S_0,Y_0,0)
    \end{equation*}
    under \eqref{LSV}.
    When $0\leq \delta \leq 1$, this is an interpolation of the local volatility price and the Bachelier price, hence is free from butterfly arbitrage opportunities.
\end{remark}

\end{samepage}

\subsection{Rough volatility models}
We now specialize Theorem~\ref{thm1} to the rough Bergomi model of
Bayer et al.~\cite{BFG}.  Let $W$ and $W^\perp$ be independent
Brownian motions, set
\begin{equation*}
 B^S_t=\rho W_t+\sqrt{1-\rho^2}\,W^\perp_t,
 \qquad \rho\in[-1,1],
\end{equation*}
and consider
\begin{equation}\label{rbergomi}
 \frac{\mathrm{d}S_t}{S_t}=\sqrt{V_t}\,\mathrm{d}B^S_t,
 \qquad
 V_t=\xi_0(t)\exp\left\{
 \nu\sqrt{2H}\int_0^t(t-s)^{H-1/2}\mathrm{d}W_s
 -\frac{\nu^2}{2}t^{2H}\right\},
\end{equation}
where $H\in(0,1/2)$, $\nu>0$, and the forward variance curve $\xi_0$
is positive and continuous at zero.  Since $\mathsf{E}[V_t]=\xi_0(t)$,
the function $v$ in Theorem~\ref{thm1} is precisely $\xi_0$.  Brownian
scaling gives
\begin{equation*}
 \frac{1}{\theta^H}\left(\frac{V_\theta}{\xi_0(\theta)}-1\right)
 \longrightarrow
 \eta:=\nu\sqrt{2H}\int_0^1(1-s)^{H-1/2}\mathrm{d}W_s.
\end{equation*}
Moreover, $\xi=B^S_1$ and $\zeta=\sup_{s\leq1}B^S_s$ in
\eqref{conv}.  Decomposing $W=\rho B^S+\sqrt{1-\rho^2}\,\widetilde W$,
where $\widetilde W$ is independent of $B^S$, shows that only the first
component contributes to the conditional expectation in $\alpha$.

The conditional Brownian maximum identities in the appendix then reduce
the coefficient to a one-dimensional integral.  With $\bar\psi$ defined in
the proof of Theorem~\ref{thm1},
\begin{equation}\label{alpha-rbergomi}
\alpha =  \alpha_{\mathrm{rB}}
 :=\frac{\rho\nu\sqrt{2H}}{2}
 \int_0^1(1-s)^{H-1/2}\bar\psi(b,k,s)\,\mathrm{d}s .
\end{equation}
For completeness, this simplification follows by first obtaining
\begin{equation*}
 \frac{\rho\nu\sqrt{2H}}{2}
 \int_0^1 \partial_3\bar\psi(b,k,s)s^{H+1/2}
 \beta_{H+1/2,-H-1/2}(1-s)\,\mathrm{d}s,
\end{equation*}
where $\beta_{a,b}(z)=\int_0^z t^{a-1}(1-t)^{b-1}\mathrm{d}t$.
Writing $a=H+1/2$, the identity
\begin{equation*}
 \beta_{a,-a}(1-s)=\frac{1}{a}\left(\frac{1-s}{s}\right)^a
\end{equation*}
and integration by parts yield \eqref{alpha-rbergomi}; the boundary terms
vanish.  Notice that at $H=1/2$, formula \eqref{gamma2} gives
$\alpha_{\mathrm{rB}}=(\rho\nu/2)\gamma(b,k)$, consistently with
Corollary~\ref{cor1}.

Combining \eqref{main} and \eqref{alpha-rbergomi}, we obtain the practical
approximation
\begin{equation}\label{rbergomi-approx}
 p^{\mathrm{UO}}\approx
 p_{\mathrm{BS}}^{\mathrm{UO}}
 \left(S_0,\sqrt{\int_0^\theta\xi_0(r)\mathrm{d}r}\right)
 +\alpha_{\mathrm{rB}}\sqrt{\xi_0(0)}S_0\theta^{H+1/2}.
\end{equation}

We examine \eqref{rbergomi-approx} with a flat forward variance curve and
the parameters
\begin{equation*}
 H=0.1,\quad \nu=1,\quad \rho=-0.9,\quad
 \xi_0(t)=0.04,\quad S_0=100,\quad B=120,\quad \theta=0.25,
\end{equation*}
for 21 equally spaced strikes $K\in[80,120]$.  The Monte Carlo benchmark
uses the common production and pilot sample sizes specified above,
and the 3R hybrid scheme~\cite{FH} on $25{,}000$ time steps over
$\theta=0.25$, so $\Delta t=10^{-5}$.  The first two singular-kernel
cells are sampled exactly, cells 3 through 10 are obtained by the
3R orthogonal projection, and the remaining convolution uses
$L^2$-optimal hybrid weights and is evaluated by FFT.
Writing $X_i=\log S_i$, the stock is discretized as
\begin{equation*}
 X_{i+1}=X_i-\tfrac12 V_i\Delta t+
 \sqrt{V_i}\bigl(\rho\Delta W_i+\sqrt{1-\rho^2}\Delta W_i^\perp\bigr).
\end{equation*}
The barrier is monitored only at grid endpoints, and the fitted
$K=B$ control is applied exactly as in the regular-model experiments.
The comparison with $12{,}500$ monitoring dates uses the same
$25{,}000$-step stock and variance paths and assesses only the
monitoring component of the numerical error.

Figure~\ref{fig:rbergomi-uop} compares Monte Carlo, the Black--Scholes
barrier price, and \eqref{rbergomi-approx}.  Across the 21 strikes,
the mean absolute pricing error decreases from $0.27126$ for
Black--Scholes to $0.10319$ for the rough Bergomi
approximation; the corresponding maximum absolute error decreases
from $0.52969$ to $0.18689$.
The largest Monte Carlo standard error is $0.00649$.
Increasing the number of monitoring dates from $12{,}500$ to
$25{,}000$ changes the controlled estimates by at most
$0.00005$, with a maximum paired standard error
of $0.00003$.
The correction captures a substantial part of the leverage effect,
while the remaining discrepancy includes approximation, sampling,
and discretization errors.

\begin{figure}[htbp]
 \centering
 \includegraphics[width=0.98\linewidth]{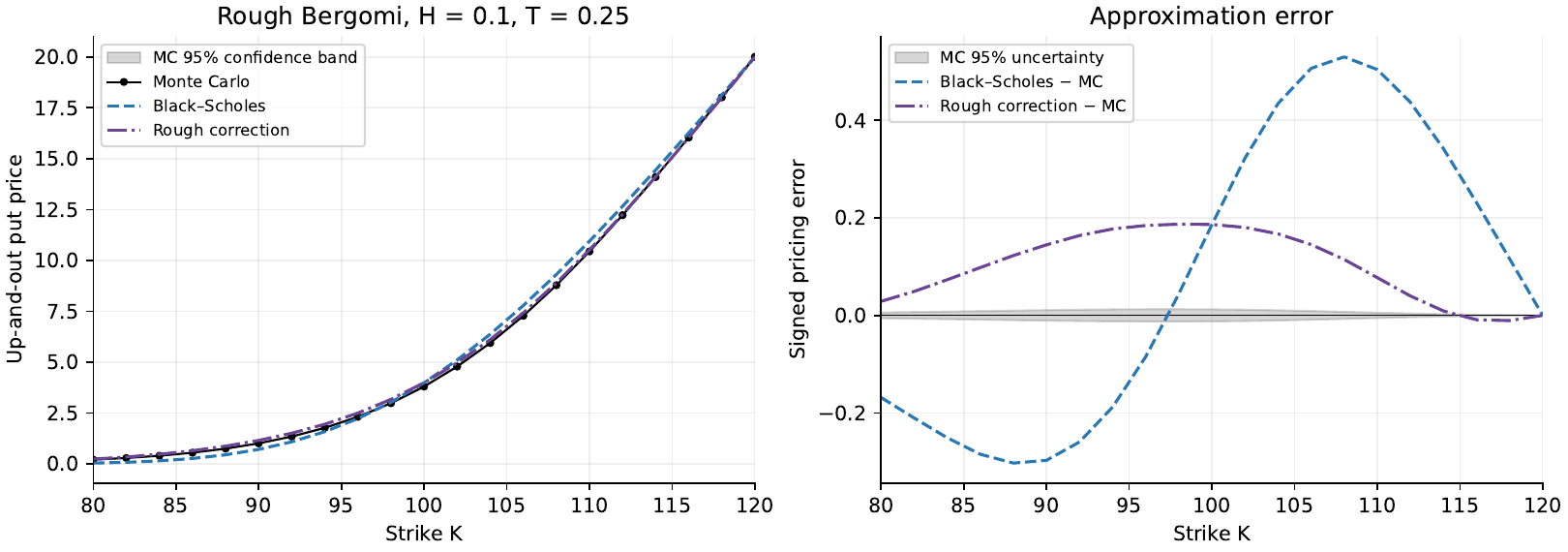}
 \caption{Rough Bergomi up-and-out put prices for $S_0=100$, $B=120$,
 $K\in[80,120]$, and $\theta=0.25$.  The shaded regions are pointwise
 95\% Monte Carlo confidence bands.  The right panel plots each
 approximation minus the 3R Monte Carlo estimate.}
 \label{fig:rbergomi-uop}
\end{figure}

\appendix
\section{Computing conditional expectations}
Let $W = \{W_t\}$ be a one-dimensional standard Brownian motion and
\begin{equation*}
    M_t = \sup_{s \in [0,t]}W_s.
\end{equation*}
Let
\begin{equation*}
    \phi(x,t) = \frac{1}{\sqrt{2\pi t}} \exp\left(-\frac{x^2}{2t} \right).
\end{equation*}
Here we give a formula to compute
\begin{equation*}
    \mathsf{E}[W_t| M_1 \leq m, W_1 = x].
\end{equation*}
By the reflection principle,
\begin{equation}\label{RP}
    \mathsf{P}[M_t \leq m, W_t \in \mathrm{d}x] = \psi(m,x,t)\mathrm{d}x,\ \ 
    \psi(m,x,t) = \phi(x,t)-\phi(2m-x,t)
\end{equation}
and so,
\begin{equation*}
\begin{split}
       & \mathsf{P}[M_1 \leq m, W_t \in \mathrm{d}y, W_1 \in \mathrm{d}x] \\
       &= \mathsf{P}[M_t \leq m, W_t \in \mathrm{d}y]
          \mathsf{P}[M_{1-t} \leq m-y, W_{1-t} \in \mathrm{d}(x-y)] \\
          &= \psi(m,y,t)\psi(m-y,x-y,1-t)\mathrm{d}y\mathrm{d}x
\end{split}
\end{equation*}
We can check that
\begin{equation*}
    \int_{-\infty}^m\psi(m,y,t)\psi(m-y,x-y,1-t)\mathrm{d}y = \psi(m,x,1).
\end{equation*}
A direct but lengthy computation yields
\begin{equation}\label{TC}
    \int_{-\infty}^m y\psi(m,y,t)\psi(m-y,x-y,1-t)\mathrm{d}y = 
    \bar{\phi}(m,x,t)-\bar{\phi}(m,2m-x,t),
\end{equation}
where
\begin{equation*}
    \bar{\phi}(m,x,t) = \phi(x,1)\left(2m-tx-2(m-tx)\Phi\left(\frac{m-tx}{\sqrt{t(1-t)}}\right)\right).
\end{equation*}
\begin{proposition}\label{prop:conditional-maximum}
For $0<t<1$ and $m>\max\{x,0\}$,
   \begin{equation*}
   \begin{split}
       \mathsf{E}[W_t| M_1 \leq m, W_1 = x] = \frac{\bar{\psi}(m,x,t)}{\psi(m,x,1)}.
   \end{split}
\end{equation*}
\end{proposition}

{\it Proof:}
By the Markov property at time $t$ and the reflection-principle density
\eqref{RP}, the conditional density of $W_t$, given
$\{M_1\leq m,W_1=x\}$, is
\begin{equation*}
 \frac{\psi(m,y,t)\psi(m-y,x-y,1-t)}{\psi(m,x,1)}
 1_{\{y<m\}}\,\mathrm{d}y.
\end{equation*}
Indeed, the denominator is the integral of the numerator by the
Chapman--Kolmogorov identity displayed above.  Hence, using \eqref{TC},
\begin{equation*}
\begin{split}
 \mathsf{E}[W_t\mid M_1\leq m,W_1=x]
 &=\frac{1}{\psi(m,x,1)}
   \int_{-\infty}^m y\psi(m,y,t)\psi(m-y,x-y,1-t)\,\mathrm{d}y\\
 &=\frac{\bar\phi(m,x,t)-\bar\phi(m,2m-x,t)}{\psi(m,x,1)}
 =\frac{\bar\psi(m,x,t)}{\psi(m,x,1)}.
\end{split}
\end{equation*}
\hfill{$\square$}

\section*{Declaration of AI use}
The author used OpenAI's ChatGPT to assist with coding, language editing and the checking of calculations and mathematical arguments. The Lean formalization of all the proofs was carried out by ChatGPT. All AI-generated suggestions were critically reviewed and, where appropriate, revised by the author, who takes full responsibility for the content of this paper.

\end{document}